\documentclass[final,3p,times,numbers,sort&compress]{elsarticle}

\usepackage{amsmath}
\usepackage{amsfonts}
\usepackage{amsmath}
\usepackage{amssymb}
\usepackage{amsthm}
\usepackage[utf8]{inputenc}
\usepackage{graphicx}
\usepackage{bm}
\usepackage{bbm}

\journal{Physica A}

\usepackage[breaklinks=true,colorlinks=true,linkcolor=blue,urlcolor=blue,citecolor=blue]{hyperref}

\newcommand{\defeq}{\mathrel{\mathop:}=}
\renewcommand{\H}{\mathcal{H}}
\newcommand{\U}{\mathcal{U}}
\renewcommand{\S}{\mathcal{S}}
\newcommand{\params}{\boldsymbol{\lambda}}

\newtheorem*{theorem}{Theorem}

\begin{document}

\begin{frontmatter}

\title{Fundamental temperature in the superstatistical description\\ of non-equilibrium steady states}

\author[cchen,unab]{Sergio Davis\corref{cor1}}
\address[cchen]{Research Center on the Intersection of Plasma Physics, Matter and Complexity (P$^2$mc),\\ Comisión Chilena de Energía Nuclear, Casilla 188-D, Santiago, Chile}
\address[unab]{Departamento de Física y Astronomía, Facultad de Ciencias Exactas, Universidad Andres Bello,\\ Sazié 2212, piso 7, 8370136, Santiago, Chile.}
\ead{sergio.davis@cchen.cl}

\cortext[cor1]{Corresponding author}

\begin{abstract}
Among the statistical mechanical frameworks able to describe systems in non-equilibrium steady states such as collisionless plasmas, self-gravitating systems and other complex systems, superstatistics have gained 
recent attention. Superstatistics postulates a superposition of canonical systems with inverse temperatures $\beta$ described by a probability distribution depending on the external conditions. Unfortunately, the 
uncertainty about $\beta$ cannot be attributed to fluctuations of a phase space function, and this suggests that the distribution of $\beta$ is purely of statistical nature and must be inferred rather than measured. 
This lack of direct observability of the superstatistical temperature then becomes a conceptual issue in need of resolution. In this work we address this issue, showing that a mapping exists between functions of the 
superstatistical temperature and functions of the recently proposed fundamental temperature, a model-dependent function of the energy, in such a way that their expectation values coincide. We illustrate the use of 
this mapping by computing the conditional distribution of inverse temperature given energy for the $q$-canonical ensemble, as well as the full inverse temperature distribution, without the use of Laplace inversion.
\end{abstract}

\end{frontmatter}

\section{Introduction}

The appropriate extension of the concept of temperature to non-equilibrium states remains an open problem in statistical mechanics. One rather elegant path towards solving this problem involves the theory of 
superstatistics~\cite{Beck2003, Beck2004}, which has been recently successful in describing plasmas\cite{Ourabah2015, Ourabah2020b, Sanchez2021}, self-gravitating systems\cite{Ourabah2020c} and other complex systems out 
of  equilibrium\cite{Denys2016, Sanchez2025}. In superstatistics, systems in non-equilibrium steady states are described as superpositions of equilibrium states at different temperatures, thus temperature in these systems 
is treated as a random variable with its own probability density. As is well known, a system in equilibrium at temperature $T$ is described by the canonical ensemble,
\begin{equation}
\label{eq:canon}
P(\bm \Gamma|\beta) = \frac{\exp\big(-\beta\H(\bm \Gamma)\big)}{Z(\beta)},
\end{equation}
which is the probability (density) of observing a microstate $\bm \Gamma$ at inverse temperature $\beta \defeq 1/(k_B T)$. Here $\H(\bm \Gamma)$ is the Hamiltonian of the system and $Z(\beta)$ is the partition function.
In superstatistics, the canonical ensemble in \eqref{eq:canon} is replaced by the joint distribution of microstates $\bm \Gamma$ and inverse temperature $\beta$, which we will write as
\begin{equation}
\label{eq:super_joint}
P(\bm \Gamma, \beta|\params) = P(\bm \Gamma|\beta)P(\beta|\params) = \frac{\exp\big(-\beta \H(\bm \Gamma)\big)}{Z(\beta)}P(\beta|\params),
\end{equation}
under external conditions that will be represented by the set of parameters $\params$. This joint distribution is given simply by the product rule of probability theory~\cite{vonDerLinden2014}, where $P(\beta|\params)$ is 
the superstatistical distribution of inverse temperature $\beta$. By integrating out the variable $\beta$, the superstatistical microstate distribution is given by
\begin{equation}
\label{eq:super_marg}
P(\bm \Gamma|\params) = \int_0^\infty d\beta\,\frac{P(\beta|\params)}{Z(\beta)}\exp\big(-\beta\H(\bm \Gamma)\big).
\end{equation}

In this way, each choice of $P(\beta|\params)$ leads to a microstate distribution different from \eqref{eq:canon}. Superstatistical states then belong to the wider class of \emph{energy steady states} (ESS), that is, 
non-equilibrium steady states where the microstate probability (density) is completely determined by the microstate energy. In other words, an ESS has a microstate distribution of the form
\begin{equation}
\label{eq:rho}
P(\bm \Gamma|\params) = \rho\big(\H(\bm \Gamma); \params\big),
\end{equation}
where $\rho$ is called the \emph{ensemble function}. By defining the \emph{superstatistical weight function} $f(\beta; \params)$ as
\begin{equation}
\label{eq:super_weight}
f(\beta; \params) \defeq \frac{P(\beta|\params)}{Z(\beta)},
\end{equation}
we directly see from \eqref{eq:super_marg} that
\begin{equation}
\label{eq:super_Laplace}
\rho(E; \params) = \int_0^\infty d\beta\,f(\beta; \params)\exp(-\beta E),
\end{equation}
i.e. the superstatistical ensemble function $\rho$ is the Laplace transform of $f$. 

Despite the success of the theory of superstatistics, there are at least two difficulties in this approach. The first issue is that not every ESS can be described using superstatistics, that is, not every 
non-negative function $\rho(E; \params)$ can be expressed as the Laplace transform of a non-negative function $f(\beta; \params)$. Notable exceptions to superstatistics include the Gaussian ensemble~\cite{Challa1988, 
Challa1988a, Johal2003, Suzuki2022}, the $q$-canonical ensemble of Tsallis statistics~\cite{Tsallis1988, Tsallis2009c} with $q < 1$ and the microcanonical ensemble.

The second is perhaps a more profound issue. It has been shown~\cite{Davis2018} that in superstatistical states, the inverse temperature $\beta$ cannot be directly measured as the value of a phase-space observable 
$B(\bm \Gamma)$. This implies that there is no function $B(\bm \Gamma)$ such that
\begin{equation}
\big<G(\beta)\big>_{\params} = \big<G(B)\big>_{\params}
\end{equation}
for every function $G$. In particular, this implies that $P(\beta|\params)$ cannot be estimated by computing a histogram of a function $B$, and instead has to be inferred indirectly. In contrast, the energy $E$ can 
always be measured directly using the Hamiltonian function $\H(\bm \Gamma)$, so for any function $G(E)$ the equality
\begin{equation}
\big<G(E)\big>_{\params} = \big<G(\H)\big>_{\params}
\end{equation}
always holds. 

In this work we provide a powerful result that compensates the lack of direct observability of $\beta$. We prove here that, for every superstatistical ESS and function $G(\beta)$, there exists a function 
$G_{\params}(\beta_F)$ such that
\begin{equation}
\big<G(\beta)\big>_{\params} = \big<G_{\params}(\beta_F)\big>_{\params}
\end{equation}
holds, where $\beta_F$ is the \emph{fundamental inverse temperature function}~\cite{Davis2019, Davis2023b}, defined for any ESS by
\begin{equation}
\label{eq:betaF_def}
\beta_F(E; \params) \defeq -\frac{\partial}{\partial E}\ln \rho(E; \params),
\end{equation}
and $G_{\params}$ is a linearly-transformed version of $G$.

\section{Properties of the fundamental temperature in superstatistics}

The main connection point between the fundamental inverse temperature $\beta_F$ in \eqref{eq:betaF_def} and the inverse temperature $\beta$ of the superstatistical formalism is the identity
\begin{equation}
\label{eq:mean_beta_E}
\big<\beta\big>_{E, \params} = \int_0^\infty d\beta\,P(\beta|E, \params)\beta = \beta_F(E; \params),
\end{equation}
where the conditional distribution of $\beta$ given $E$ is
\begin{equation}
\label{eq:pbeta_E}
P(\beta|E, \params) = \frac{f(\beta; \params)\exp(-\beta E)}{\rho(E; \params)}.
\end{equation}

\noindent
The joint distribution of energy $E$ and inverse temperature $\beta$ is obtained from \eqref{eq:super_joint} as
\begin{equation}
P(E, \beta|\params) = f(\beta; \params)\exp(-\beta E)\Omega(E),
\end{equation}
and from it, the conditional distribution of $\beta$ given $E$ is
\begin{equation}
P(\beta|E, \params) = \frac{P(E, \beta|\params)}{\int_0^\infty d\beta\,P(E, \beta|\params)} = \frac{f(\beta; \params)\exp(-\beta E)}{\int_0^\infty d\beta\,f(\beta; \params)\exp(-\beta E)}
\end{equation}
which is \eqref{eq:pbeta_E}. Its moments are given by
\begin{equation}
\label{eq:betamom_E}
\big<\beta^n\big>_{E, \params} \defeq \int_0^\infty d\beta\,P(\beta|E, \params)\beta^n = \frac{(-1)^n}{\rho(E; \params)}\frac{\partial^n \rho(E; \params)}{\partial E^n}
\end{equation}
for all integers $n \geq 0$ and all allowed energies $E$, and, as was shown in Ref.~\cite{Davis2026b}.

\begin{figure}[b!]
\begin{center}
\includegraphics[width=0.33\textwidth]{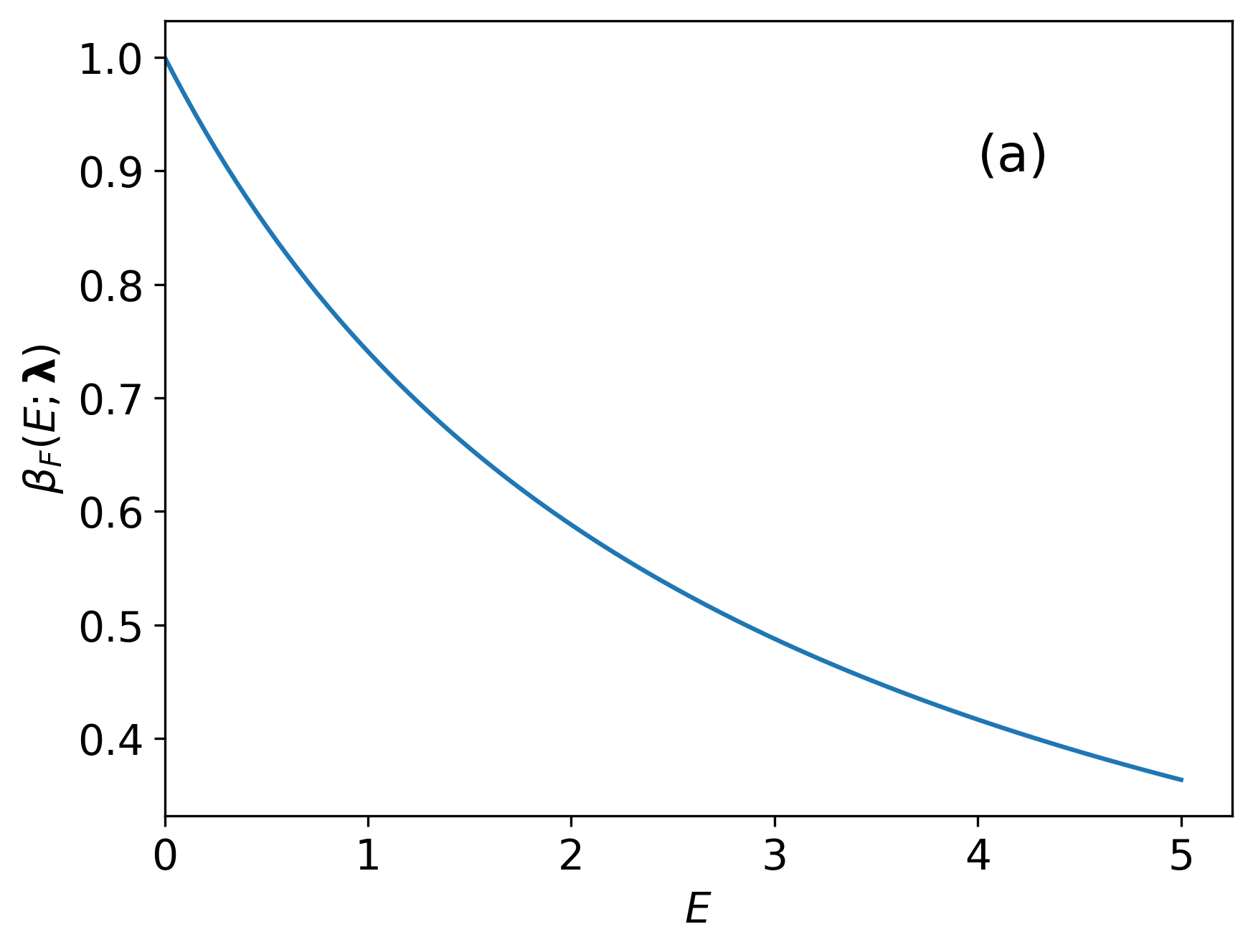}
\includegraphics[width=0.33\textwidth]{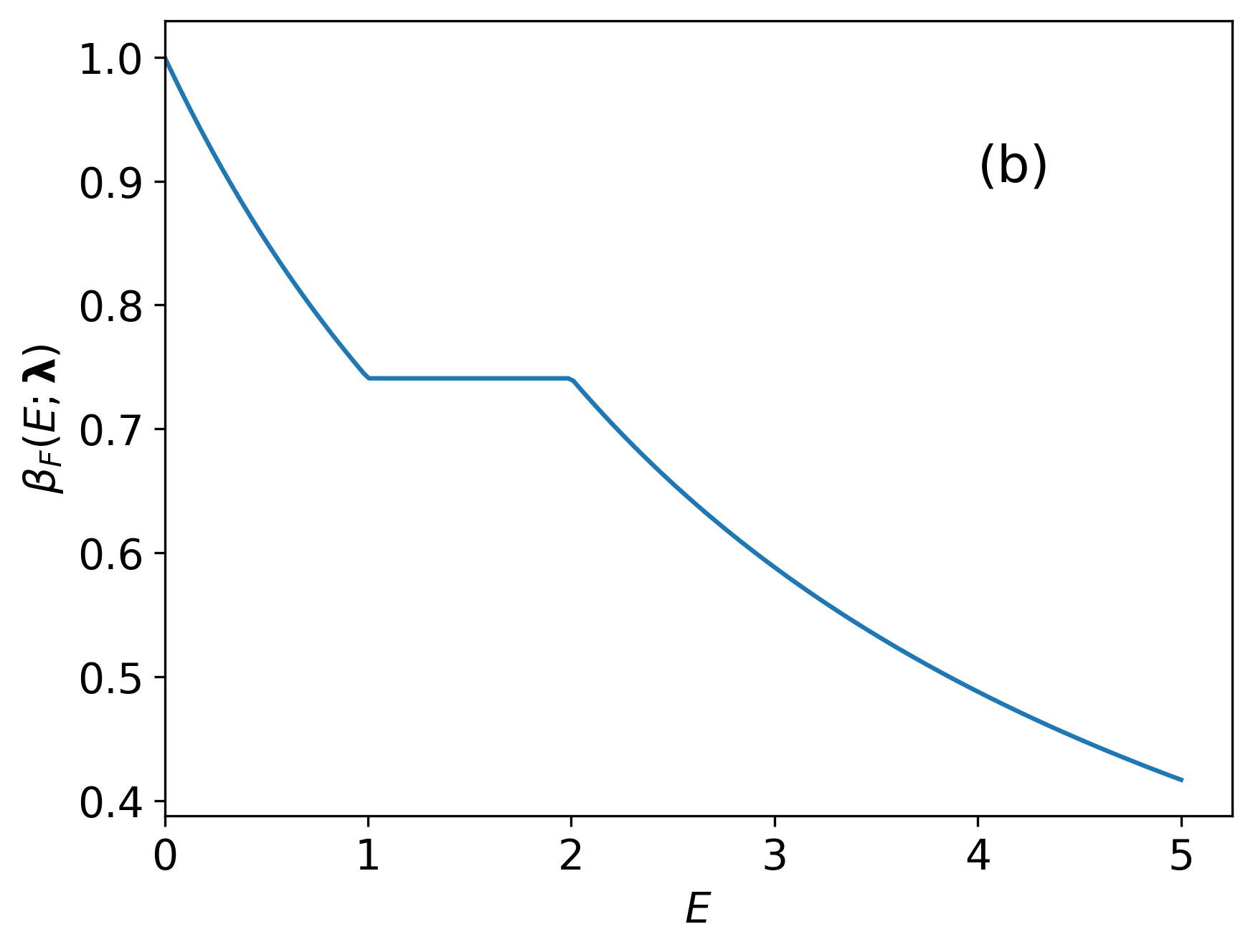}
\includegraphics[width=0.33\textwidth]{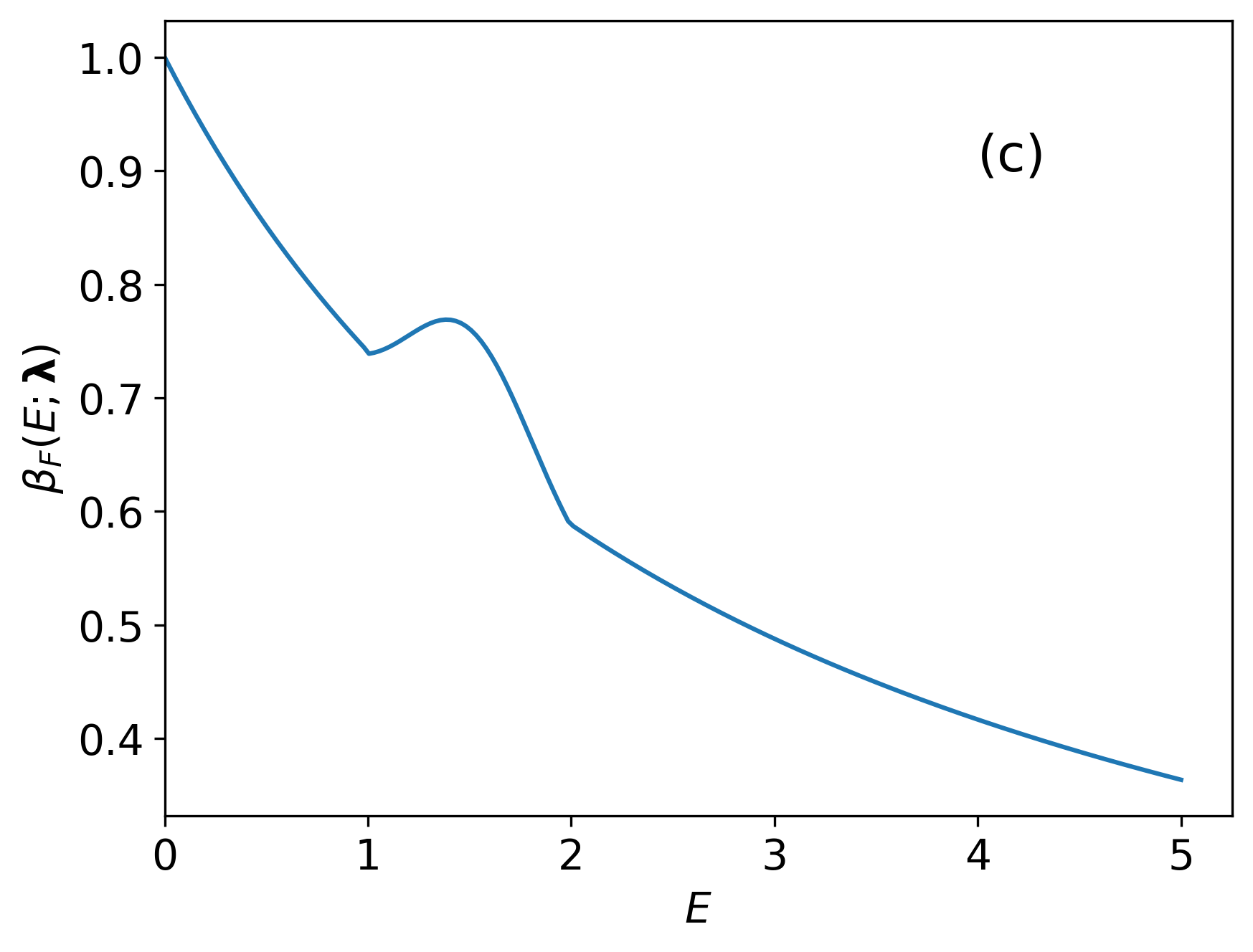}
\end{center}
\caption{Three examples of functions $\beta_F$. In the case (a), the function is strictly decreasing, therefore ${\beta_F}' < 0$ and, furthermore, $\beta_F$ is invertible. The case (b) has an interval where ${\beta_F}' = 0$, 
so $\beta_F$ is decreasing but not strictly, and therefore not invertible. In both cases (a) and (b), the value of $\beta_F$ determines the value of its derivative, in agreement with \eqref{eq:betaF_der}. Finally, in case (c) the 
derivative ${\beta_F}'$ does change sign, breaking the requirement in \eqref{eq:super_nec} and thus is not compatible with superstatistics.}
\label{fig:three_betaF}
\end{figure}

\vspace{5pt}
\noindent
From \eqref{eq:betamom_E} with $n = 2$ we obtain the second moment of $\beta$ given $E$ as
\begin{equation}
\label{eq:mean2_beta_E}
\big<\beta^2\big>_{E, \params} = \beta_F(E; \params)^2 - {\beta_F}'(E; \params),
\end{equation}
and combining \eqref{eq:mean_beta_E} and \eqref{eq:mean2_beta_E} we can write the variance of \eqref{eq:pbeta_E} as
\begin{equation}
\label{eq:var_beta_E}
\big<(\delta \beta)^2\big>_{E, \params} = -{\beta_F}'(E; \params).
\end{equation}

\noindent
Now, because the left-hand side of \eqref{eq:var_beta_E} is a variance, it is non-negative by definition, so we must have
\begin{equation}
\label{eq:super_nec}
{\beta_F}'(E; \params) \leq 0
\end{equation}
for all allowed energies $E$. This means that $\beta_F$ must be a decreasing function of $E$, and this condition implies that we can always construct a function $A(\beta_F)$ such that
\begin{equation}
\label{eq:betaF_der}
{\beta_F}' = A(\beta_F),
\end{equation}
that is, the value of $\beta_F$ by itself fixes the value of the derivative ${\beta_F}'$, without making reference to the energy $E$. Please note that the condition \eqref{eq:betaF_der} can be fulfilled even for functions 
$\beta_F$ that are not invertible, thus \eqref{eq:betaF_der} is true even if there are intervals of energy where $\beta_F$ is constant. These cases are depicted in Fig.~\ref{fig:three_betaF}.

The ordinary differential equation (ODE) in \eqref{eq:betaF_der} belongs to the class of \emph{autonomous} ODEs. Therefore we can understand superstatistical models as solutions of autonomous ODEs. From \eqref{eq:betaF_der} 
we can show by induction that the higher-order derivatives of $\beta_F$ also depend only on $\beta_F$, that is, there is a sequence of functions $A_1, A_2, A_3, \ldots$ such that
\begin{equation}
\label{eq:betaF_dern}
{\beta_F}^{(n)} = A_n(\beta_F)
\end{equation}
for all integer $n \geq 1$. Clearly \eqref{eq:betaF_dern} for $n = 1$ reduces to \eqref{eq:betaF_der} with $A_1 = A$, so it is true. Assuming \eqref{eq:betaF_dern} for $n$ we easily verify that \eqref{eq:betaF_dern} for 
$n+1$ follows, as
\begin{equation}
{\beta_F}^{(n+1)} = \frac{\partial {\beta_F}^{(n)}}{\partial E} = \frac{\partial A_n(\beta_F)}{\partial E} = A_n'(\beta_F){\beta_F}' = A_n'(\beta_F) A(\beta_F) = A_{n+1}(\beta_F),
\end{equation}
where we have identified
\begin{equation}
A_{n+1}(\beta_F) = A_n'(\beta_F) A(\beta_F).
\end{equation}

\section{A new perspective on the observability of the inverse temperature}

Now we can proceed with the central result of this work, namely, the proof that for any superstatistical model, the conditional distribution $P(\beta|E, \params)$ in \eqref{eq:pbeta_E} only depends on $E$ through the 
function $\beta_F(E; \params)$, even when this function is not in general invertible. In order to show why this is true, let us first consider the moment-generating function (MGF) of $P(\beta|E, \params)$, defined by
\begin{equation}
\label{eq:R_def}
R_{\params}(t; E) \defeq \big<\exp(-\beta t)\big>_{E, \params} = \sum_{n = 0}^\infty \frac{(-t)^n}{n!}\,\big<\beta^n\big>_{E, \params},
\end{equation}
which is such that
\begin{equation}
\label{eq:mom_from_R}
\big<\beta^n\big>_{E, \params} = (-1)^n\left[\frac{\partial^n R_{\params}(t; E)}{\partial t^n}\right]_{t = 0}.
\end{equation}

\noindent
From its definition, the function $R_{\params}(t; E)$ is also equal to the Laplace transform of $P(\beta|E, \params)$,
\begin{equation}
\label{eq:R_Laplace}
\int_0^\infty d\beta\,P(\beta|E, \params)\exp(-\beta t) = R_{\params}(t; E),
\end{equation}
and, replacing \eqref{eq:betamom_E} into the right-hand side of \eqref{eq:R_def}, it can be written as the \emph{ratio function}
\begin{equation}
\label{eq:R_ratio}
R_{\params}(t; E) = \frac{1}{\rho(E; \params)}\sum_{n = 0}^\infty \frac{t^n}{n!}\frac{\partial^n \rho(E; \params)}{\partial E^n} = \frac{\rho(E+t; \params)}{\rho(E; \params)}.
\end{equation}

\noindent
Taking logarithm of the rightmost-hand side of \eqref{eq:R_ratio} and expanding it in a Taylor series, we have
\begin{equation}
\ln \rho(E+t; \params)- \ln \rho(E; \params)
= \sum_{n=1}^\infty \frac{t^n}{n!}\frac{\partial^n}{\partial E^n}\ln \rho(E; \params)
= -\sum_{n=1}^\infty \frac{t^n}{n!}{\beta_F}^{(n-1)}(E; \params)
= -\sum_{n=1}^\infty \frac{t^n}{n!} A_{n-1}\big(\beta_F(E; \params)\big),
\end{equation}
which means
\begin{equation}
R_{\params}(t; E) = \exp\left(-\sum_{n=1}^\infty \frac{t^n}{n!} A_{n-1}\big(\beta_F(E; \params)\big)\right),
\end{equation}
thus $R_{\params}(t; E)$ depends on $E$ only through the function $\beta_F(E; \params)$. From here on, for simplicity we will let $\beta'$ denote one of the allowed values of $\beta_F$, and using this notation
we can replace $R_{\params}(t; E)$ by a new function $\mathcal{R}_{\params}(t, \beta')$ given by
\begin{equation}
\mathcal{R}_{\params}(t; \beta') = \exp\left(-\sum_{n=1}^\infty \frac{t^n}{n!} A_{n-1}(\beta')\right),
\end{equation}
such that
\begin{equation}
\label{eq:newR}
R_{\params}(t; E) = \mathcal{R}_{\params}\big(t; \beta_F(E; \params)\big).
\end{equation}

\noindent
Replacing \eqref{eq:newR} into \eqref{eq:R_Laplace}, we obtain
\begin{equation}
\label{eq:Rp_Laplace}
\int_0^\infty d\beta\,P(\beta|E, \params)\exp(-\beta t) = \mathcal{R}_{\params}\big(t; \beta_F(E; \params)\big)
\end{equation}
and this in turn implies that $P(\beta|E, \params)$ can only depend on $E$ through $\beta_F(E; \params)$. Therefore we can write
\begin{equation}
P(\beta|E, \params) = \mathcal{P}\big(\beta; \beta_F(E; \params), \params\big)
\end{equation}
with $\mathcal{P}$ a function whose role we will clarify shortly. In fact, the joint distribution of $\beta$ and $\beta'$ can be computed as
\begin{equation}
\begin{split}
P(\beta, \beta'|\params) & = \int_{-\infty}^{\infty} dE\,P(\beta, E|\params)\delta\big(\beta_F(E; \params)-\beta'\big) \\
& = \int_{-\infty}^{\infty} dE\,P(E|\params)P(\beta|E, \params)\delta\big(\beta_F(E; \params)-\beta'\big) \\
& = \mathcal{P}(\beta; \beta', \params)\int_{-\infty}^{\infty} dE\,P(E|\params)\delta\big(\beta_F(E; \params)-\beta'\big) \\
& = \mathcal{P}(\beta; \beta', \params)P(\beta'|\params)
\end{split}
\end{equation}
and it follows that the conditional distribution of $\beta$ given $\beta'$ is
\begin{equation}
P(\beta|\beta', \params) = \frac{P(\beta, \beta'|\params)}{P(\beta'|\params)} = \mathcal{P}(\beta; \beta', \params).
\end{equation}

Therefore, without ambiguity, we can encode the fact that $P(\beta|E, \params)$ depends on $E$ only through $\beta_F(E; \params)$ into the shortcut notation
\begin{equation}
\label{eq:master}
P(\beta|E, \params) = P(\beta|\beta', \params)
\end{equation}
provided we interpret $\beta'$ as the value of $\beta_F$ corresponding to the energy $E$. This is an important result for superstatistics, as it implies that the only information about the energy that is relevant to 
$\beta$ is the value of $\beta_F$. From \eqref{eq:master} we can derive our main result. 

\begin{theorem}
Let $G$ be any function of the superstatistical inverse temperature $\beta$ such that its expectation value is well defined. Then it is possible to define a transformed function $G_{\params}$ of the fundamental inverse 
temperature $\beta_F$ such that
\begin{equation}
\label{eq:G_equality}
\big<G(\beta)\big>_{\params} = \big<G_{\params}(\beta_F)\big>_{\params}
\end{equation}
holds. Furthermore, the transformed function $G_{\params}$ is defined by
\begin{equation}
\label{eq:Gp_def}
G_{\params}(\beta') \defeq \big<G(\beta)\big>_{\beta', \params} = \int_0^\infty d\beta\,G(\beta)\,P(\beta|\beta', \params),
\end{equation}
i.e., it is the conditional expectation of $G(\beta)$ given a fixed value of $\beta_F$.

\end{theorem}

\begin{proof}
We compute the expectation of $G$ given $E$ and $\params$ as
\begin{equation}
\big<G(\beta)\big>_{E, \params} = \int_0^\infty d\beta\,G(\beta)P(\beta|E, \params)
= \int_0^\infty d\beta\,G(\beta)\Big[P(\beta|\beta', \params)\Big]_{\beta' = \beta_F(E; \params)}
= \left[\int_0^\infty d\beta\,G(\beta)\,P(\beta|\beta', \params)\right]_{\beta' = \beta_F(E; \params)}
\end{equation}
where we have used \eqref{eq:master} to replace $P(\beta|E, \params)$ by $P(\beta|\beta', \params)$. Using the definition of $G_{\params}$ in \eqref{eq:Gp_def}, we have
\begin{equation}
\label{eq:G_equality_E}
\big<G(\beta)\big>_{E, \params} = G_{\params}\big(\beta_F(E; \params)\big),
\end{equation}
and by taking expectation of $\big<G(\beta)\big>_{E, \params}$ given $\params$ and using \eqref{eq:G_equality_E}, we have
\begin{equation}
\label{eq:G_equality_pre}
\big<G(\beta)\big>_{\params} = \int_{-\infty}^{\infty}\hspace{-7pt}dE\,P(E|\params)\big<G(\beta)\big>_{E, \params} = \int_{-\infty}^{\infty}\hspace{-7pt}dE\,P(E|\params)G_{\params}\big(\beta_F(E; \params)\big).
\end{equation}

\noindent 
Finally, introducing a factor of 1 as
\begin{equation}
1 = \int_0^\infty d\beta'\,\delta\big(\beta_F(E; \params)-\beta'\big)
\end{equation}
in the rightmost-hand side of \eqref{eq:G_equality_pre} and changing the order of integration, we obtain
\begin{equation}
\begin{split}
\big<G(\beta)\big>_{\params} & = \int_{-\infty}^{\infty}\hspace{-7pt}dE\,P(E|\params)\left[\int_0^\infty d\beta'\,\delta\big(\beta_F(E; \params)-\beta'\big)\right]G_{\params}\big(\beta_F(E; \params)\big) \\
& = \int_0^\infty\hspace{-7pt}d\beta'\,G_{\params}(\beta')\int_{-\infty}^{\infty}\hspace{-7pt}dE\,P(E|\params)\delta\big(\beta_F(E; \params)-\beta'\big) \\
& = \int_0^\infty d\beta'\,G_{\params}(\beta')P(\beta'|\params) = \big<G_{\params}(\beta_F)\big>_{\params},
\end{split}
\end{equation}
which is \eqref{eq:G_equality}.
\end{proof}

\section{Conditional distribution of inverse temperature in the $q$-canonical ensemble}

The $q$-canonical ensemble is a generalization of \eqref{eq:canon} commonly used in the context of Tsallis' non-extensive statistical mechanics~\cite{Tsallis2009c}, where the parameter $q$ is known as the 
entropic index. Within the superstatistical framework, $q \geq 1$ is required. The ensemble function for the $q$-canonical ensemble in that case is 
\begin{equation}
\label{eq:qcanon}
\rho(E; q, \beta_0) = \frac{1}{Z_q(\beta_0)}\Big[1 + (q-1)\beta_0 E\Big]^{\frac{1}{1-q}}
\end{equation}
with $\beta_0 > 0$. The corresponding fundamental inverse temperature function is
\begin{equation}
\label{eq:qcanon_betaF}
\beta_F(E; q, \beta_0) = \frac{\beta_0}{1 + (q-1)\beta_0 E},
\end{equation}
so that $0 \leq \beta_F \leq \beta_0$, therefore the parameter $\beta_0$ can be understood as the upper bound for the fundamental inverse temperature. The derivative of $\beta_F$ is given by
\begin{equation}
{\beta_F}'(E; q, \beta_0) = (1-q)\beta_F(E; q, \beta_0)^2,
\end{equation}
and we see that the function $A(\beta)$ in \eqref{eq:betaF_der} must be
\begin{equation}
\label{eq:qcanon_A}
A(\beta) = (1-q)\beta^2.
\end{equation}

As an application of the formalism developed in the previous section, we will compute the conditional distribution $P(\beta|E, \params)$ for the $q$-canonical ensemble without using the Laplace inversion of $\rho(E; \params)$ 
to obtain $f(\beta; \params)$. In order to do this, we will first use \eqref{eq:mom_from_R} to compute the conditional moments of $\beta$ given $\beta'$, as
\begin{equation}
\label{eq:betan_Rq}
\big<\beta^n\big>_{\beta', q, \beta_0} = (-1)^n\left[\frac{\partial^n \mathcal{R}_q(t; \beta')}{\partial t^n}\right]_{t = 0}.
\end{equation}

\noindent
Using \eqref{eq:qcanon} we can compute the ratio function $R_{q, \beta_0}(t; E)$ in \eqref{eq:R_ratio} as
\begin{equation}
R_{q, \beta_0}(t; E) = \left[\frac{1 + (q-1)\beta_0\big(E+t\big)}{1 + (q-1)\beta_0 E}\right]^{\frac{1}{1-q}}
= \left[1 + \frac{1 + (q-1)\beta_0t}{1 + (q-1)\beta_0 E}\right]^{\frac{1}{1-q}},
\end{equation}
which can be written as a function of $\beta_F$ as
\begin{equation}
\label{eq:qcanon_Rp}
\mathcal{R}_q(t; \beta') = \Big[1 + \beta'(q-1)t\Big]^{\frac{1}{1-q}} = \exp_q\big(-\beta'\,t\big)
\end{equation}
where $\exp_q$ is the $q$-exponential function~\cite{Naudts2011},
\begin{equation}
\exp_q(x) \defeq \big[1 + (1-q)x\big]_+^{\frac{1}{1-q}}.
\end{equation}

Here we immediately note that, in $\mathcal{R}_q$, the explicit dependence on $\beta_0$ disappears. Interestingly, by taking expectation of \eqref{eq:Rp_Laplace} given $(q, \beta_0)$ and replacing \eqref{eq:qcanon_Rp} 
it follows that
\begin{equation}
\Big<\exp(-\beta t)\Big>_{q, \beta_0} = \Big<\exp_q\big(-\beta_F t\big)\Big>_{q, \beta_0}.
\end{equation}

\noindent
By replacing \eqref{eq:qcanon_Rp} into \eqref{eq:betan_Rq} and using the $n$-th derivative of the $q$-exponential function
\begin{equation}
\frac{\partial^n}{\partial x^n}\exp_q(x) = \exp_q(x)^{1 - n(1-q)}\,(q-1)^n\,\frac{\Gamma\left(n + \frac{1}{q-1}\right)}{\Gamma\left(\frac{1}{q-1}\right)},
\end{equation}
we obtain
\begin{equation}
\label{eq:qcanon_mom}
\big<\beta^n\big>_{\beta', q} = (q-1)^n \frac{\Gamma\left(n + \frac{1}{q-1}\right)}{\Gamma\left(\frac{1}{q-1}\right)}\,(\beta')^n
\end{equation}

\noindent
Here we note that the $q$-canonical ensemble is the only superstatistical model where
\begin{equation}
\label{eq:qcanon_cond}
\big<\beta^n\big>_{E, \params} = f_n(\params)\,(\beta')^n
\end{equation}
for all integers $n \geq 1$. In fact, From \eqref{eq:qcanon_cond} with $n = 2$ we have that
\begin{equation}
\big<\beta^2\big>_{E, \params} = f_2(\params)(\beta')^2,
\end{equation}
and, replacing it into \eqref{eq:mean2_beta_E},
\begin{equation}
{\beta_F}' = \big[1 - f_2(\params)\big](\beta')^2,
\end{equation}
which is \eqref{eq:qcanon_A}, whose unique solution is \eqref{eq:qcanon_betaF}, i.e. the $q$-canonical ensemble.

Now we will use the conditional moments in \eqref{eq:qcanon_mom} to recover $P(\beta|E, \params)$ without using the inversion of a Laplace transform. For this we consider a function $G(\beta)$ with power series
\begin{equation}
G(\beta) = \sum_{n = 0}^\infty C_n \beta^n,
\end{equation}
whose transformed function $G_q$ is given by
\begin{equation}
\label{eq:Gq_series}
G_q(\beta') = \sum_{n=0}^\infty C_n(q-1)^n \frac{\Gamma\left(n + \frac{1}{q-1}\right)}{\Gamma\left(\frac{1}{q-1}\right)}\,(\beta')^n.
\end{equation}

\noindent
We can then rewrite \eqref{eq:Gq_series} in terms of the function $G$ by using the integral form of the gamma function, as
\begin{equation}
\label{eq:Gq_integral}
\begin{split}
G_q(\beta') & = \frac{1}{\Gamma\left(\frac{1}{q-1}\right)} = \sum_{n=0}^\infty C_n\,\Gamma\Big(n + \frac{1}{q-1}\Big) \big[\beta'(q-1)\big]^n 
= \frac{1}{\Gamma\left(\frac{1}{q-1}\right)}\sum_{n = 0}^\infty C_n\left[\int_0^\infty ds\,s^{n + \frac{1}{q-1}-1}\exp(-s)\right] \big[\beta'(q-1)\big]^n \\
& = \int_0^\infty ds\,\frac{s^{\frac{1}{q-1}-1}\exp(-s)}{\Gamma\left(\frac{1}{q-1}\right)}\sum_{n = 0}^\infty C_n \big[\beta'(q-1)s\big]^n 
= \int_0^\infty ds\,\frac{s^{\frac{1}{q-1}-1}\exp(-s)}{\Gamma\left(\frac{1}{q-1}\right)} G\big(\beta'(q-1)s\big),
\end{split}
\end{equation}
and performing the change of variables from $s$ to $\beta \defeq \beta'(q-1)s$ we can write \eqref{eq:Gq_integral} as 
\begin{equation}
G_q(\beta') = \int_0^\infty\hspace{-7pt}d\beta\left[\frac{1}{\beta'(q-1)\,\Gamma\left(\frac{1}{q-1}\right)}\left(\frac{\beta}{\beta'(q-1)}\right)^{\frac{1}{q-1}-1}\hspace{-12pt}
\exp\left(-\frac{\beta}{\beta'(q-1)}\right)\right]G(\beta).
\end{equation}
 
It follows by comparison with \eqref{eq:Gp_def} that the function in brackets must be the conditional distribution $P(\beta|\beta', q, \beta_0)$. That is, we have
\begin{equation}
\label{eq:qcanon_betadist}
P(\beta|\beta', q) = \frac{1}{\beta'(q-1)\,\Gamma\left(\frac{1}{q-1}\right)}\exp\left(-\frac{\beta}{\beta'(q-1)}\right)\left(\frac{\beta}{\beta'(q-1)}\right)^{\frac{1}{q-1}-1},
\end{equation}
result which is verified in \ref{app:laplace} using the Laplace inversion of \eqref{eq:super_Laplace}. This is a gamma distribution with mean and variance
\begin{subequations}
\begin{align}
\big<\beta\big>_{\beta', q, \beta_0} & = \beta', \\
\big<(\delta \beta)^2\big>_{\beta', q, \beta_0} & = (q-1)(\beta')^2,
\end{align}
\end{subequations}
in agreement with \eqref{eq:mean_beta_E} and \eqref{eq:var_beta_E}, respectively.

\section{Inverse temperature distribution for the $q$-canonical ensemble}

In this section we will use the distribution of $\beta$ given $\beta'$ obtained in \eqref{eq:qcanon_betadist} to compute $P(\beta|q, \beta_0)$ in the case of a system with constant heat capacity $C_E = \alpha k_B$. 
The density of states for this kind of system is given by
\begin{equation}
\label{eq:DOS_alpha}
\Omega(E) = \omega_0\,E^\alpha
\end{equation}
and the energy distribution $P(E|q, \beta_0)$ can be computed by replacing \eqref{eq:qcanon} and \eqref{eq:DOS_alpha} into 
\begin{equation}
P(E|\params) = \rho(E; \params)\Omega(E).
\end{equation}

\noindent
The result can be written as
\begin{equation}
\label{eq:qcanon_Edist}
P(\varepsilon|q, \beta_0) = \frac{\big(1 + \varepsilon\big)^{\frac{1}{1-q}}\,\varepsilon^\alpha}{B\left(\alpha+1, \frac{1}{q-1}-\alpha-1\right)},
\end{equation}
which is an inverted beta distribution for the reduced variable $\varepsilon \defeq \beta_0(q-1)E$. Now invoking the marginalization rule over the variable $\beta'$ we have
\begin{equation}
\label{eq:pbeta_marg}
P(\beta|q, \beta_0) = \int_0^\infty d\beta'\,P(\beta, \beta'|q, \beta_0) = \int_0^\infty d\beta'\,P(\beta|\beta', q)P(\beta'|q, \beta_0),
\end{equation}
and we can obtain the distribution $P(\beta'|q, \beta_0)$ from $P(\varepsilon|q, \beta_0)$ in \eqref{eq:qcanon_Edist} as
\begin{equation}
\label{eq:pbetap_pre}
P(\beta'|q, \beta_0) = \left<\delta\left(\frac{\beta_0}{1 + \varepsilon} - \beta'\right)\right>_{q, \beta_0}
= \frac{1}{\beta_0}\Big<(1+\varepsilon)^2\,\delta\big(\varepsilon-\varepsilon_F(\beta')\big)\Big>_{q, \beta_0}
= \frac{\beta_0}{(\beta')^2} P\big(\varepsilon = \varepsilon_F(\beta')\big|q, \beta_0\big)
\end{equation}
where $\varepsilon_F$ is given by
\begin{equation}
\label{eq:EF}
\varepsilon_F(\beta') = \frac{\beta_0}{\beta'} - 1.
\end{equation}

\noindent
Replacing \eqref{eq:qcanon_Edist} and \eqref{eq:EF} into \eqref{eq:pbetap_pre} then yields
\begin{equation}
P(\beta'|q, \beta_0) = \frac{1}{\beta_0 B\left(\alpha+1, \frac{1}{q-1}-\alpha-1\right)}\,\left(\frac{\beta_0}{\beta'}\right)^{\frac{1}{1-q}+2}\left(\frac{\beta_0}{\beta'}-1\right)^{\alpha}, \qquad 0 \leq \beta' \leq \beta_0,
\end{equation}
which means, using \eqref{eq:pbeta_marg}, that
\begin{equation}
P(\beta|q, \beta_0) = \frac{\big[\beta_0(q-1)\big]^{1 - \frac{1}{q-1}}\,\beta^{\frac{1}{q-1}-1}}{\Gamma(\alpha+1)\Gamma\left(\frac{1}{q-1}-\alpha-1\right)}
\int_0^\infty\hspace{-10pt} \frac{d\beta'}{(q-1)(\beta')^2}\exp\left(-\frac{\beta}{\beta'(q-1)}\right)\left(\frac{\beta_0}{\beta'}-1\right)^{\alpha},
\end{equation}
then we finally obtain
\begin{equation}
P(\beta|q, \beta_0) = \frac{1}{\beta_0(q-1)\,\Gamma\left(\frac{1}{q-1}-\alpha-1\right)}\exp\left(-\frac{\beta}{\beta_0(q-1)}\right)\left(\frac{\beta}{\beta_0(q-1)}\right)^{\frac{1}{q-1}-\alpha-2}.
\end{equation}

\noindent
This is a gamma distribution with mean and variance given by
\begin{subequations}
\begin{align}
\beta_S = \big<\beta\big>_{q, \beta_0} & = \beta_0\Big(1+(1-q)(\alpha+1)\Big), \\
\U = \big<(\delta \beta)^2\big>_{q, \beta_0} & = (\beta_0)^2 (q-1)\Big(1 + (1-q)(\alpha+1)\Big),
\end{align}
\end{subequations}
in agreement with Eqs. (83) and (85) of Ref.~\cite{Davis2022b}. The reduced inverse temperature covariance, defined by
\begin{equation}
u \defeq \frac{\U}{(\beta_S)^2},
\end{equation}
depends only on $q$, and is given by
\begin{equation}
u = \frac{q-1}{1 + (1-q)(\alpha+1)}.
\end{equation}

By replacing $\beta_0$ and $q$ in terms of $u$ and $\beta_S$ we can simplify the distribution of inverse temperature, which now reads
\begin{equation}
P(\beta|u, \beta_S) = \frac{1}{u\beta_S\,\Gamma\left(\frac{1}{u}\right)}\exp\left(-\frac{\beta}{u\beta_S}\right)\left(\frac{\beta}{u\beta_S}\right)^{\frac{1}{u}-1}.
\end{equation}

We emphasize here that this distribution leads to the $q$-canonical ensemble only in systems where the partition function $Z(\beta)$ is such that the weight function $f(\beta; q, \beta_0)$ in \eqref{eq:super_weight} 
is proportional to a gamma distribution.

\section{A generalization of the entropic index $q$ for superstatistics}

\begin{figure}[b!]
\begin{center}
\includegraphics[width=0.7\textwidth]{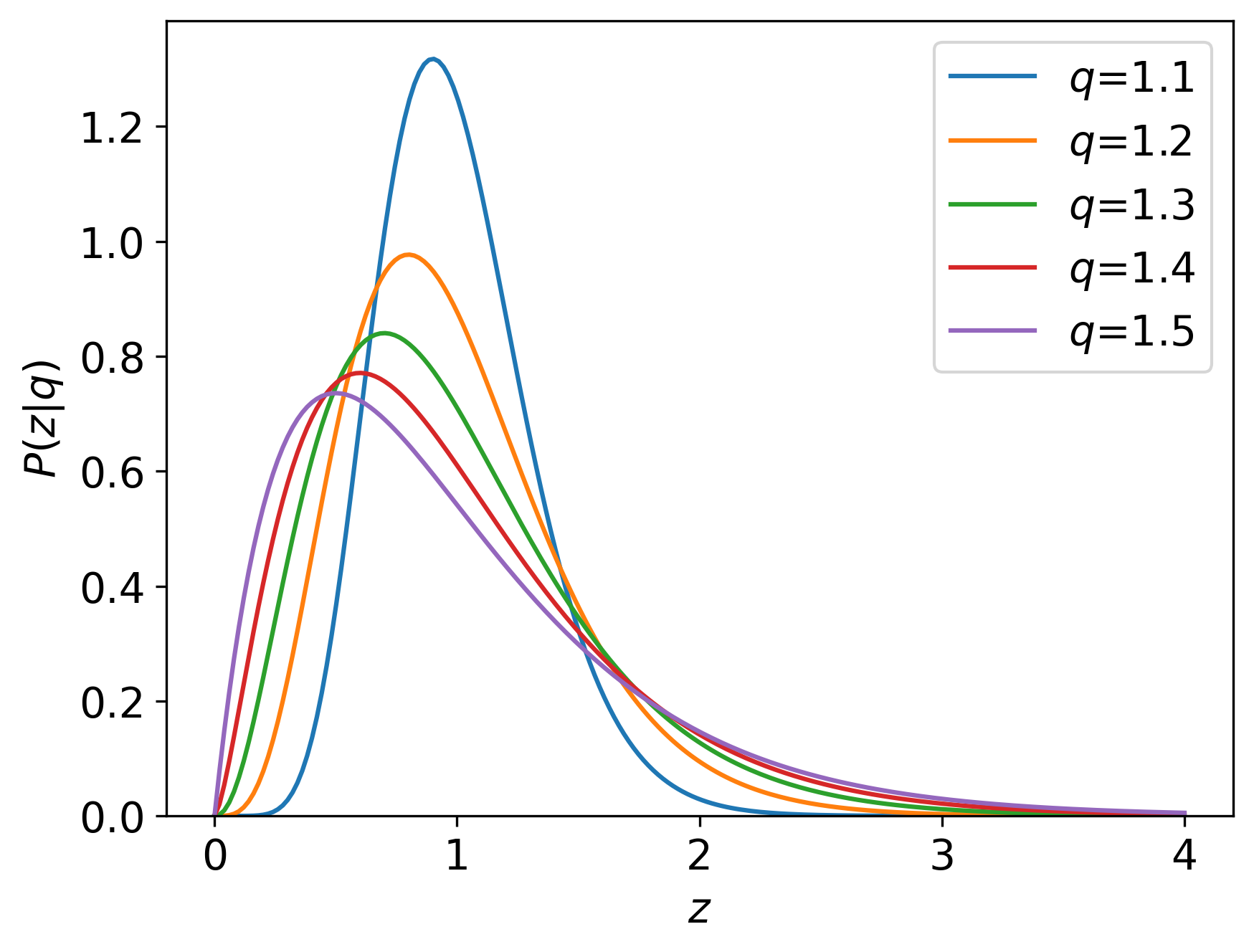}
\end{center}
\caption{Distribution of $z \defeq \beta/\beta_F$ for the $q$-canonical ensemble, for different values of $q$.}
\label{fig:probz}
\end{figure}

\noindent
In the $q$-canonical ensemble, as seen from \eqref{eq:qcanon_betadist}, the reduced inverse temperature
\begin{equation}
z \defeq \frac{\beta}{\beta'}
\end{equation}
has a universal statistical distribution which depends solely on $q$. In fact, we can explicitly compute the distribution of $z$ given $\beta'$ and $q$ as
\begin{equation}
P(z|\beta', q) = \int_0^\infty d\beta\,P(\beta|\beta', q)\delta\left(z-\frac{\beta}{\beta'}\right)
= \frac{1}{(q-1)\Gamma\left(\frac{1}{q-1}\right)}\exp\left(-\frac{z}{q-1}\right)\left(\frac{z}{q-1}\right)^{\frac{1}{q-1}-1},
\end{equation}
which does not actually depend on $\beta'$. Therefore we also have that
\begin{equation}
\label{eq:probz}
P(z|q) = \frac{1}{(q-1)\Gamma\left(\frac{1}{q-1}\right)}\exp\left(-\frac{z}{q-1}\right)\left(\frac{z}{q-1}\right)^{\frac{1}{q-1}-1},
\end{equation}
which is a gamma distribution, as shown in Fig.~\ref{fig:probz}. The mean and variance of $z$ for a given $q$ are
\begin{subequations}
\begin{align}
\big<z\big>_q & = 1, \\
\big<(\delta z)^2\big>_q & = q-1,
\end{align}
\end{subequations}
respectively, and from them we obtain
\begin{equation}
\label{eq:z2_q}
\big<z^2\big>_q = q.
\end{equation}

\noindent
Here \eqref{eq:z2_q} suggests a generalization of Tsallis' entropic index $q$, namely
\begin{equation}
Q(\params) \defeq \left<\left(\frac{\beta}{\beta_F}\right)^2\right>_{\params}
\end{equation}
for any superstatistical model, such that $Q(q, \beta_0) = q$ for the $q$-canonical ensemble. Moreover, from \eqref{eq:mean2_beta_E} it follows that
\begin{equation}
\left<\left(\frac{\beta}{\beta_F}\right)^2\right>_{E, \params} =1 - \frac{{\beta_F}'(E; \params)}{\beta_F(E; \params)^2}
\end{equation}
and by taking expectation given $\params$,
\begin{equation}
Q(\params) = \left<\left(\frac{\beta}{\beta_F}\right)^2\right>_{\params} = 1 - \left<\frac{{\beta_F}'}{(\beta_F)^2}\right>_{\params}.
\end{equation}

\noindent
In this way, from the inequality \eqref{eq:super_nec} it follows that
\begin{equation}
Q(\params) \geq 1
\end{equation}
for any superstatistical model. Because the distribution $P(z|q)$ is universal for $q$-canonical systems, we can use its entropy
\begin{equation}
\S_z(q) \defeq -\int_0^\infty dz\,P(z|q)\ln P(z|q)
\end{equation}
as the basis for an entropic prior~\cite{Caticha2004, Neumann2007}, of the form
\begin{equation}
\label{eq:entropic_prior}
P(q|\varnothing) = \frac{1}{\zeta}\exp\big(\S_z(q)\big).
\end{equation}

\noindent
Replacing \eqref{eq:probz} we have
\begin{equation}
\begin{split}
\S_z(q) & = \ln\,(q-1) + \ln \Gamma\left(\frac{1}{q-1}\right) + \int_0^\infty ds\,\frac{\exp(-s)s^{\frac{1}{q-1}-1}}{\Gamma\left(\frac{1}{q-1}\right)}\left[s - \Big(\frac{1}{q-1}-1\Big)\ln s\right] \\
& = \ln\,(q-1) + \ln \Gamma\left(\frac{1}{q-1}\right) + \frac{1}{q-1} - \Big(\frac{1}{q-1}-1\Big)\,\psi\left(\frac{1}{q-1}\right),
\end{split}
\end{equation}
where $\psi$ is the digamma function, and replacing $\S_z(q)$ into \eqref{eq:entropic_prior} we see that the resulting entropic prior
\begin{equation}
\label{eq:prior_q}
P(q|\varnothing) = \frac{(q-1)}{\zeta}\Gamma\left(\frac{1}{q-1}\right)\exp\left(\frac{1}{q-1}\left[1 + (q-2)\,\psi\Big(\frac{1}{q-1}\Big)\right]\right), \qquad q \geq 1
\end{equation}
is normalizable for $q \geq 1$ with $\zeta \approx $ 8.68848. Interestingly, this prior has a well-defined mode at $q = 2$, as shown in Fig.~\ref{fig:entropic}.

\begin{figure}[h!]
\begin{center}
\includegraphics[width=0.6\textwidth]{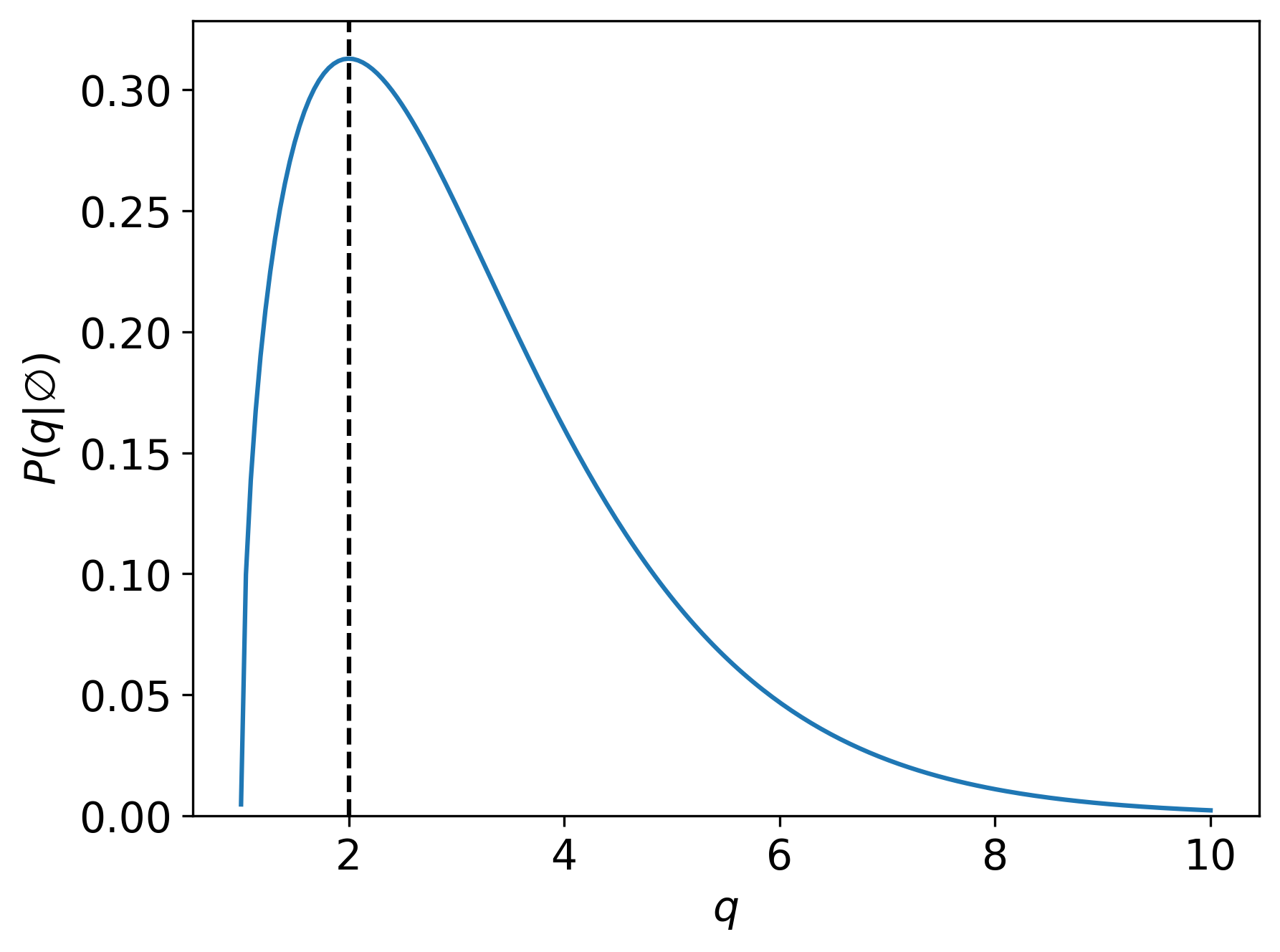}
\end{center}
\caption{Entropic prior $P(q|\varnothing)$ as defined in \eqref{eq:prior_q}.}
\label{fig:entropic}
\end{figure}

\section{Concluding remarks}

We have established a central property of the fundamental inverse temperature functions $\beta_F$ of superstatistical models, namely that they are solutions of autonomous ordinary differential equations, and this 
implies that higher-order derivatives of $\beta_F$ are functions of $\beta_F$ as well. From this property, it follows that a mapping exists between functions of the inverse temperature and functions of the fundamental 
inverse temperature, such that their expectation values are preserved. This result may provide future insights into the nature of the superstatistical and fundamental inverse temperatures, as well as practical 
theoretical tools for the analysis of superstatistical systems.

\section*{Acknowledgments}

\noindent
Funding from ANID FONDECYT 1220651 grant is gratefully acknowledged.

\newpage
\appendix
\section{Verification of the conditional distribution of inverse temperature in the $q$-canonical ensemble}
\label{app:laplace}

The conditional distribution $P(\beta|\beta_F, q, \beta_0)$ in \eqref{eq:qcanon_betadist} for the $q$-canonical ensemble can be verified by the Laplace inversion of
\begin{equation}
\int_0^\infty d\beta\,f(\beta; q, \beta_0)\exp(-\beta E) = \frac{1}{Z_q(\beta_0)}\Big[1 + (q-1)\beta_0 E\Big]^{\frac{1}{1-q}},
\end{equation}
which is \eqref{eq:super_Laplace} for the ensemble function $\rho(E; q, \beta_0)$ in \eqref{eq:qcanon}. From the gamma integral
\begin{equation}
\int_0^\infty dt\,t^{c-1}\,\exp(-at)\exp(-st) = \Gamma(c)(a + s)^{-c}
\end{equation}
we obtain the corresponding weight function $f(\beta; q, \beta_0)$ as
\begin{equation}
\label{eq:f_qcanon}
f(\beta; q, \beta_0) = \frac{1}{Z_q(\beta_0)\beta_0(q-1)\,\Gamma\left(\frac{1}{q-1}\right)}\exp\left(-\frac{\beta}{\beta_0(q-1)}\right)\left(\frac{\beta}{\beta_0(q-1)}\right)^{\frac{1}{q-1}-1}.
\end{equation}

\noindent
Replacing \eqref{eq:qcanon} and \eqref{eq:f_qcanon} into \eqref{eq:pbeta_E}, we have
\begin{equation}
\label{eq:pbeta_E_qcanon}
P(\beta|E, q, \beta_0) = \frac{\Big[1 + \beta_0(q-1)E\Big]^{\frac{1}{q-1}}}{\beta_0(q-1)\,\Gamma\left(\frac{1}{q-1}\right)}\exp\left(-\frac{\beta\Big[1 + \beta_0(q-1)E\Big]}{\beta_0(q-1)}\right)\left(\frac{\beta}{\beta_0(q-1)}\right)^{\frac{1}{q-1}-1},
\end{equation}
and by using \eqref{eq:qcanon_betaF} in the form
\begin{equation}
1 + \beta_0(q-1)E = \frac{\beta_0}{\beta'}
\end{equation}
we obtain
\begin{equation}
P(\beta|E, q, \beta_0) = \frac{1}{\beta_0(q-1)\Gamma\left(\frac{1}{q-1}\right)}\left(\frac{\beta_0}{\beta'}\right)^{\frac{1}{q-1}}\exp\left(-\frac{\beta}{\beta'(q-1)}\right)\left(\frac{\beta}{\beta_0(q-1)}\right)^{\frac{1}{q-1}-1}
\end{equation}
which is $P(\beta|\beta', q)$ in \eqref{eq:qcanon_betadist}.

\bibliography{fts}
\bibliographystyle{unsrt}

\end{document}